\documentclass[twoside]{article}
\usepackage{qic,epsfig}

\usepackage[utf8]{inputenc}
\usepackage[english]{babel}
\usepackage[T1]{fontenc}
\usepackage{amsmath,amssymb,amsthm,nicefrac,hyperref,braket,qcircuit,array}
\newcolumntype{?}{!{\vrule width 1.2pt}}

\newtheorem{lmm}{Lemma}
\newtheorem{thm}{Proposition}
\newtheorem{cor}{Corollary}

\DeclareMathOperator{\spn}{\operatorname{span}}

\newcommand{\multiprepareC}[2]{*+<1em,.9em>{\hphantom{#2}}\save[0,0].[#1,0];p\save !C
  *{#2},p+RU+<0em,0em>;+LU+<+.8em,0em> **\dir{-}\restore\save +RD;+RU **\dir{-}\restore\save
  +RD;+LD+<.8em,0em> **\dir{-} \restore\save +LD+<0em,.8em>;+LU-<0em,.8em> **\dir{-} \restore \POS
  !UL*!UL{\cir<.9em>{u_r}};!DL*!DL{\cir<.9em>{l_u}}\restore}

\begin{document}
\setlength{\textheight}{8.0truein}

\runninghead{The signaling dimension in generalized probabilistic theories}
            {Michele Dall'Arno, Alessandro Tosini, and Francesco Buscemi}

\normalsize\textlineskip
\thispagestyle{empty}
\setcounter{page}{411}

%\copyrightheading{Vol.}{No.}{Year}{Page Nos.}
%\copyrightheading{24}{5\&6}{2024}{0411--0424}

\vspace*{0.88truein}

\alphfootnote

\fpage{411}

\centerline{\bf 
%The signaling dimension in generalized probabilistic theories
THE SIGNALING DIMENSION IN GENERALIZED PROBABILISTIC THEORIES}

\vspace*{0.37truein}

\centerline{\footnotesize 
%Michele Dall'Arno
MICHELE DALL'ARNO}

\vspace*{0.015truein}

\centerline{\footnotesize\it Department  of Computer Science
  and  Engineering,  Toyohashi   University  of  Technology,
  Japan}

\baselineskip=10pt

\centerline{\footnotesize\it michele.dallarno.mv@tut.jp}

\vspace*{10pt}

\centerline{\footnotesize 
%Alessandro Tosini
ALESSANDRO TOSINI}

\vspace*{0.015truein}

\centerline{\footnotesize\it Physics  Department, University
  of Pavia, Italy}

\baselineskip=10pt

\centerline{\footnotesize\it alessandro.tosini@unipv.it}

\vspace*{10pt}

\centerline{\footnotesize 
%Francesco Buscemi
FRANCESCO BUSCEMI}

\vspace*{0.015truein}

\centerline{\footnotesize\it Graduate School of Informatics,
  Nagoya University, Japan}

\baselineskip=10pt

\centerline{\footnotesize\it buscemi@nagoya-u.jp}

\vspace*{0.225truein}

%\publisher{November 21, 2023}{March 31, 2024}

\vspace*{0.21truein}

\abstracts{
  The  signaling  dimension  of   a  given  physical  system
  quantifies  the minimum  dimension of  a classical  system
  required to reproduce all input/output correlations of the
  given system. Thus, unlike other dimension measures - such
  as the dimension of the linear space or the maximum number
  of (jointly or pairwise)  perfectly discriminable states -
  which examine  the correlation  space only along  a single
  direction, the signaling dimension  does not depend on the
  arbitrary choice  of a specific operational  task. In this
  sense, the signaling dimension summarizes the structure of
  the  entire set  of  input/output correlations  consistent
  with  a given  system in  a single  scalar quantity.   For
  quantum  theory, it  was  recently proved  by Frenkel  and
  Weiner in  a seminal  result that the  signaling dimension
  coincides with the Hilbert space dimension.
}{
  Here, we  derive analytical and algorithmic  techniques to
  compute the  signaling dimension  for any given  system of
  any given generalized probabilistic theory.  We prove that
  it  suffices   to  consider  extremal   measurements  with
  ray-extremal effects, and we  bound the number of elements
  of   any  such   measurement  in   terms  of   the  linear
  dimension. For  systems with  a finite number  of extremal
  effects,  we  recast  the problem  of  characterizing  the
  extremal  measurements with  ray-extremal  effects as  the
  problem of  deriving the vertex description  of a polytope
  given  its face  description,  which  can be  conveniently
  solved  by   standard  techniques   such  as   the  double
  description  algorithm.   For  each such  measurement,  we
  recast  the computation  of the  signaling dimension  as a
  linear program, and we  propose a combinatorial branch and
  bound algorithm to reduce its  size.  We apply our results
  to  derive  the  extremal measurements  with  ray-extremal
  effects of  a composition of  two square bits  (or squits)
  and  prove that  their signaling  dimension is  five, even
  though each squit has a signaling dimension equal to two.
}{}

\vspace*{10pt}

%\vspace*{3pt}
%\communicate{to be filled by the Editorial}

\vspace*{1pt}\textlineskip

\section{Introduction}
\noindent
Generalized probabilistic theories~\cite{Spe07,Bar07,CDP11},
of which classical and quantum theories represent particular
instances,   represent   the   most   general   mathematical
description   of   a   physical   theory.    A   generalized
probabilistic theory  is specified  in terms  of the  set of
states and  the set of  effects of its systems.   The former
correspond to the admissible  preparation procedures for the
system, while  the latter represent the  admissible building
blocks of measurements. A rule to generate composite systems
and  their  allowed  dynamics   can  be  specified,  further
enriching the structure of the theory.

Within  any  given   generalized  probabilistic  theory,  an
operationally relevant problem is  to quantify the dimension
of     any     given      system.      Several     different
quantifiers~\cite{HW12,  BKLS14}   of  dimension   had  been
introduced, including  the dimension of the  linear space of
the states and effects of  the system, or the maximum number
of (jointly or pairwise) perfectly distinguishable states of
the system.  On the one  hand, the former quantifier lacks a
direct   operational   interpretation   in  terms   of   the
input/output correlations  achievable by the system;  on the
other  hand,  the  latter  quantifier is  dependent  on  the
arbitrary  choice   of  an   operational  task,   and  hence
investigates the space of  input/output correlations along a
single direction only.

In    stark    contrast    with    that,    the    signaling
dimension~\cite{Dal22}  does  not  depend on  the  arbitrary
choice of a specific  operational task, and hence summarizes
the structure of the entire set of input/output correlations
that is  consistent with a  given system in a  single scalar
quantity. Formally,  the signaling  dimension of  the system
quantifies the minimum dimension of any simulating classical
system, that is, any classical system that can reproduce all
the input-output correlations of the given system.

For quantum  theory, it  was recently  proved~\cite{FW15} by
Frenkel  and  Weiner in  a  groundbreaking  result that  the
signaling  dimension   coincides  with  the   Hilbert  space
dimension.   Subsequently,  the  problem  of  computing  the
signaling   dimension  in   different  contexts   has  drawn
considerable    attention,     for    instance     in    the
case~\cite{HSGB18, DBB20,  DC21, CGDJ21, FW22}  of arbitrary
classical  and   quantum  channels,   as  well  as   in  the
case~\cite{DBTBV17,  MK18,  HKL20,  Fre22}  of  channels  in
generalized probabilistic theories.

In  this   work,  we   derive  analytical   and  algorithmic
techniques to compute the  signaling dimension for any given
system of  any given  generalized probabilistic  theory.  We
split the  problem of  computing the signaling  dimension in
two  steps:   i)  the   characterization  of   the  extremal
measurements of  any given system,  that is relevant  in its
own right for optimization problems other than the signaling
dimension, and  ii) the actual computation  of the signaling
dimension   given    the   characterization    of   extremal
measurements.

Concerning the first step, we prove that, when computing the
signaling  dimension,  it  suffices  to  consider  extremal
measurements with ray-extremal effects, and we show that the
number  of  elements of  any  such  a measurement  is  upper
bounded by  the dimension of  the linear space.  For systems
whose set of admissible effects is a polytope, we recast the
problem  of characterizing  the  extremal measurements  with
ray-extremal effects as the problem of deriving the vertices
description of a polytope given its faces description.  Such
a problem can be  conveniently solved by standard techniques
such as the double description algorithm.

Regarding the  second step, that is,  the actual computation
of  the signaling  dimension given  the characterization  of
extremal measurements,  we recast it  as a series  of linear
programs, one  for each extremal measurement.   We propose a
combinatorial, branch  and bound algorithm to  reduce such a
size  and  make it  practically  tractable.   We provide  an
implementation~\cite{Dal23}  of the  such  an algorithm,  as
well  as of  the other  algorithms discussed  in this  work,
released under a free software license.

As  a  running example  thorough  our  work, we  consider  a
composition of  two systems. For  each of such  systems, the
set of  admissible states is geometrically  represented by a
square, hence such systems are also known as square bits, or
squits.  Square  systems have been originally  introduced as
an   implementation   of  the   Popescu-Rohrlich~\cite{PR94,
  BLMP05,  Bar07, DT10}  correlations.   In  that case,  the
composition  includes all  the  entangled states  consistent
with  the  squit  local   structure,  leaving  no  room  for
entangled   effects,   a    trade-off   first   noticed   in
Ref~\cite{SB10}.  However, alternative composition rules can
be considered with  a richer set of  entangled effects, thus
allowing   for  a   richer   characterization  of   extremal
measurement and  a non-trivial computation of  the signaling
dimension.

Incidentally, we show that alternative compositions, such as
the one  considered in Ref.~\cite{Jan12},  are inconsistent,
that   is,  they   contain   well-formed  experiments   that
nonetheless  give   rise  to  negative   probabilities.   We
classify all the consistent  composition rules of two squits
and focus on the instance considered in Ref.~\cite{DBTBV17},
a  ``dual  version''  of  the  Popescu-Rohrlich  boxes  that
includes  all  possible  entangled effects  but  only  local
preparations. We apply the present algorithmic techniques to
derive the  extremal measurements with  ray-extremal effects
of such a model and prove that its signaling dimension five.
Incidentally, this shows the tightness of the lower bound in
Ref.~\cite{DBTBV17} for  the signaling  dimension of  such a
model.

The     paper    is     structured    as     follows.     In
Section~\ref{sec:formalization} we formalize  the problem of
computing  the   signaling  dimension  as   an  optimization
problem.  In Section~\ref{sec:composition}  we introduce our
running  example  by   discussing  the  completely  positive
compositions      of      two     square      bits.       In
Section~\ref{sec:measurements}  we   derive  analytical  and
algorithmic results on the  characterization of the extremal
measurements   with  ray-extremal   effects  of   any  given
system. In Section~\ref{sec:signalingdimension} we provide a
combinatorial,  branch and  bound algorithm  for the  exact,
closed-form  computation   of  the  signal   dimension.   We
summarize   our   results   and  discuss   possible   future
developments in Section~\ref{sec:conclusion}.

\section{Formalization}
\noindent
In this  section we  introduce the object  of study  of this
work,  that is,  the  signaling dimension,  and the  running
example given by the composition of two squit systems.

\subsection{Signaling dimension}
\label{sec:formalization}
\noindent
A physical  system $S$ of linear  dimension $\ell ( S  ) \in
\mathbb{N}$  can be  represented  by  a pair  $(\mathcal{S},
\mathcal{E})$, where $\mathcal{S} \subseteq \mathbb{R}^\ell$
and $\mathcal{E} \subseteq \mathbb{R}^\ell$  are the sets of
admissible   states   and    effects,   respectively.    The
probability  of measuring  the  effect  $e \in  \mathcal{E}$
given  the state  $\omega \in  \mathcal{S}$ is  given by  $e
\cdot \omega$, and the effect $\overline{e}$ that gives unit
probability for any deterministic preparation is called unit
effect.     This     effect    is    unique     in    causal
theories~\cite{CDP11},  and it  corresponds to  the identity
operator in the quantum case.  For any given system $S$, let
$\mathcal{P}^{m    \to    n}_S$    denote   the    set    of
$m$-input/$n$-output  conditional probability  distributions
that can be generated by  system $S$ with shared randomness.
That is, $p$  is an element of $\mathcal{P}^{m  \to n}_S$ if
and  only  if  there exists  states  $\{  \omega_{x|\lambda}
\}_{x, \lambda}$  and measurements $\{  E_{y|\lambda} \}_{y,
  \lambda}$ such that
\begin{align*}
  p_{y|x}  =  \sum_{\lambda} q_{\lambda}  \omega_{x|\lambda}
  \cdot E_{y|\lambda},
\end{align*}
for   some   probability    distribution   $\{   q_{\lambda}
\}_{\lambda}$.              The            \textit{signaling
  dimension}~\cite{DBTBV17} $\kappa$ of a  system $S$ is the
minimum dimension $d$  of a classical system  $C_d$ that can
reproduce the  input/output correlations attainable  by $S$,
that is
\begin{align*}
  \kappa  \left(  S \right)  :=  \min_{d  \in \mathbb{N}}  d
  \qquad  \textrm{s.t.}   \qquad  \mathcal{P}^{m   \to  n}_S
  \subseteq \mathcal{P}^{m \to n}_{C_d},
\end{align*}
for  any  $m,  n  \in  \mathbb{N}$.  In  particular,  for  a
classical   system  $C_d$   of  dimension   $d$,  the   sets
$\mathcal{S}$  and $\mathcal{E}$  of  admissible states  and
effects,  respectively,  are  known  to  be  represented  by
regular $d-1$ simplices in a linear space of dimension $\ell
( C_d )  = d$. For instance, for the  (classical) bit, trit,
and  quart,  the states  are  represented  by a  segment,  a
triangle, and  a tetrahedron,  respectively, and so  are the
effects.

The       signaling      dimension       was      originally
introduced~\cite{DBTBV17}    in    connection    with    the
no-hyper\-signaling  principle,  which  is  a  scaling  rule
stipulating that the signaling  dimension of the composition
of any  given systems cannot  be larger than the  product of
the signaling dimension of each  system. In other words, the
no-hypersignaling   principle   constraints  the   time-like
correlations that  any given system can  exhibit, and allows
to  rule out  as  unphysical even  systems whose  space-like
correlations  are instead  compatible with  quantum or  even
classical  theory.  In  particular,  the  fact that  quantum
theory  satisfies such  a  scaling rule  follows  only as  a
consequence of the  aforementioned recent result~\cite{FW15}
by Frenkel and Weiner.

\subsection{Compositions of square bits}
\label{sec:composition}
\noindent
Here we  introduce a toy  model theory, i.e. the  square bit
(or  squit)  and its  compositions,  that  will serve  as  a
running example  across all this  work.  The squit  system S
has  linear  dimension  $\ell(S)  = 3$,  namely  its  states
$\omega$  and  effects  $e$  are  described  by  vectors  in
$\mathbb{R}^3$.    We   now    specify   the   convex   sets
$\mathcal{S}$ and $\mathcal{E}$.  The extremal states of the
squit (that is,  the four vertices of the  square) are given
by the  four vectors  $\{ \omega_k  = U_k^s  \omega_0 \}_k$,
with $\omega_0 =  (1, 0, 1)^T$.  Here, $\{  U_k^s \}_{k, s}$
are the  reversible transformations of the  system (that is,
the  symmetries  of  the  square  or,  more  generally,  the
transformation  that  leave  the set  of  admissible  states
invariant, as  opposed to those transformations  that shrink
it) given by the following rotations and reflections
\begin{align*}
  U_k^s =
  \begin{pmatrix}
    \cos \frac{\pi k}2 & -s \sin \frac{\pi k}2 & 0 \\
    \sin \frac{\pi k}2 & s \cos \frac{\pi k}2 & 0 \\
    0 & 0 & 1
  \end{pmatrix},
\end{align*}
with $k \in  \{0, 1, 2, 3\}$  and $s = \pm  1$.

By explicit  computation, the  effect space  dual to  such a
state space  (that is, the set  of effects $e$ such  that $e
\cdot  \omega \ge  0$) is  given (up  to a  positive scaling
factor) by  the convex hull  of $\{  e_k = U_k^s  e_0 \}_k$,
with $e_0 = (1, 1, 1)^T$.

Any  composition  of  two squits  necessarily  includes  the
factorized extremal states and effects, that is
\begin{align}
  \label{eq:states}
  \Omega_{4i+j} := \omega_i \otimes \omega_j,
\end{align}
and
\begin{align}
  \label{eq:effects}
  E_{4i+j} := e_i \otimes e_j,
\end{align}
respectively, where $i, j \in \{ 0, 1, 2, 3 \}$, so that the
model  has  to  include   at  least  the  bipartites  states
$\Omega_0,  \dots  \Omega_{15}$  and the  bipartite  effects
$E_0, \dots  E_{15}$. Additionally, by  explicit computation
the   set   of   effects   dual   to   Eq.~\eqref{eq:states}
includes~\cite{DT10,  SB10}  the   eight  entangled  effects
$E_{16}$ up to $E_{23}$ given by the columns of
\begin{align*}
  \begin{pmatrix}
    -1 & -1 & 1 & 1 &-1 & 1 & 1 &-1\\
    1 &-1 &-1 & 1 & 1 & 1 &-1 &-1\\
    0 & 0 & 0 & 0 & 0 & 0 & 0 & 0\\
    1 &-1 &-1 & 1 &-1 &-1 & 1 & 1\\
    1 & 1 &-1 &-1 &-1 & 1 & 1 &-1\\
    0 & 0 & 0 & 0 & 0 & 0 & 0 & 0\\
    0 & 0 & 0 & 0 & 0 & 0 & 0 & 0\\
    0 & 0 & 0 & 0 & 0 & 0 & 0 & 0\\
    1 & 1 & 1 & 1 & 1 & 1 & 1 & 1
  \end{pmatrix},
\end{align*}
respectively.    Analogously,  the   state  space   dual  to
Eq.~\eqref{eq:effects} includes~\cite{DT10,  SB10} the eight
entangled states $\Omega_{16}$ up to $\Omega_{23}$ given (up
to a positive scaling factor) by the columns of
\begin{align*}
  \begin{pmatrix}
   -1 & -1 & 1 & 1 &-1 & 1 & 1 &-1\\
   1 &-1& -1 & 1 &-1 &-1 & 1 & 1\\
   0 & 0 & 0 & 0 & 0 & 0 & 0 & 0\\
   1 &-1 &-1 & 1 & 1 & 1 &-1 &-1\\
   1 & 1 &-1 &-1 &-1 & 1 & 1 &-1\\
   0 & 0 & 0 & 0 & 0 & 0 & 0 & 0\\
   0 & 0 & 0 & 0 & 0 & 0 & 0 & 0\\
   0 & 0 & 0 & 0 & 0 & 0 & 0 & 0\\
   2 & 2 & 2 & 2 & 2 & 2 & 2 & 2
  \end{pmatrix},
\end{align*}
respectively. Notice that the  bipartite system $S\otimes S$
has linear dimension $\ell{(S\otimes S)} = \ell{(S)}^2 = 9$.

However,  it  is  well  known  that not  all  of  the  eight
entangled  states and  the  eight entangled  effects can  be
included  in the  same composite  system. Indeed,  composite
systems  must be  completely  positive, that  is, they  must
generate  non-negative probabilities  when connected  in any
possible  way allowed  by the  (reversible) dynamics  of the
system.  In general, to find  the reversible dynamics of the
composite system (that  is, the symmetries of  their sets of
admissible states  and effects),  one can use  the algorithm
introduced  in the  appendix  of  Ref.~\cite{DBK23}. In  the
particular   case  of   two   squits,  it   is  known   from
Ref.~\cite{GMCD10}  that  the   only  non-trivial  bipartite
reversible dynamics is the swap operator.

In Ref.~\cite{Jan12}, Janotta  considered a theory including
the   following   four  entangled   states:   $\Omega_{16}$,
$\Omega_{18}$,   $\Omega_{22}$,   and   $\Omega_{23}$   [see
  Eqs.~(10), (11),  (12), and (9)  therein], as well  as the
following four  entangled effects (up to  a positive scaling
factor):  $E_{17}$, $E_{19}$,  $E_{20}$,  and $E_{21}$  [see
  Eqs.~(18), (19),  (17), and (20) therein].   Such a theory
is not even positive, and therefore not completely positive,
that is, it includes  events whose probability of occurrence
is  strictly negative.   As an  example, take  the following
event,  obtained by  wiring (the  $\operatorname{Swap}$ gate
represents   the   swap   operator)  the   entangled   state
$\Omega_{23}$  with  the  entangled  effect  $E_{20}$  (both
included  in  Janotta's  model), and  whose  probability  is
strictly negative:
\begin{align*}
  \begin{aligned}
    \Qcircuit          @C=.75em          @R=.75em          {
      \multiprepareC{1}{\Omega_{22}}                       &
      \multigate{1}{\operatorname{Swap}}                   &
      \multimeasureD{1}{E_{20}}\\  \pureghost{\Omega_{22}} &
      \ghost{\operatorname{Swap}} & \ghost{E_{20}} }
  \end{aligned} < 0.
\end{align*}

To find the completely  positive compositions of two squits,
we first observe that the following relations hold
\begin{align}
  \label{eq:transpose1}
  \operatorname{Swap} \Omega_{20} = \Omega_{23},\\
  \label{eq:transpose2}
  \operatorname{Swap} \Omega_{21} = \Omega_{22},\\
  \label{eq:transpose3}
  \operatorname{Swap} E_{20} = E_{23},\\
  \label{eq:transpose4}
  \operatorname{Swap} E_{21} = E_{22}.
\end{align}
Moreover, any such a composition must satisfy
\begin{align}
  \label{eq:positivity}
  \begin{aligned}
    \Qcircuit      @C=1em      @R=.7em     @!       R      {
      \multiprepareC{1}{\Omega_x} &  \qw &  \gate{U^{n_0}} &
      \qw&        \qw        &        \multimeasureD{3}{E_y}
      \\  \pureghost{\Omega_x}   &  \qw   &  \qw  &   \qw  &
      \multimeasureD{1}{E_y}        &        \pureghost{E_y}
      \\ \multiprepareC{1}{\Omega_x} &  \qw & \gate{U^{n_1}}
      &\qw      &     \ghost{E_y}      &     \pureghost{E_y}
      \\ \pureghost{\Omega_x} & \qw & \gate{U^{n_2}} & \qw &
      \qw & \ghost{E_y} }
  \end{aligned} \ge 0,
\end{align}
where $U$ is given by
\begin{align*}
  U \propto
  \begin{aligned}
    \Qcircuit @C=1em @R=.7em @! R {
      & \multimeasureD{1}{E_y} \\
      \multiprepareC{1}{\Omega_x} & \ghost{E_y}  \\
      \pureghost{\Omega_x} & \qw}
  \end{aligned} \;,
\end{align*}
and  $n_i =  0,  1$ with $i \in \{ 0, 1, 2 \}$.

By    explicit   computation,    Eqs.~\eqref{eq:transpose1},
\eqref{eq:transpose2},                \eqref{eq:transpose3},
\eqref{eq:transpose4},  and~\ref{eq:positivity} above  imply
that the only possible compositions are:
\begin{description}
\item[PR  model] all  eight entangled  states, no  entangled
  effect, free  local dynamics, so called  since it produces
  Popescu-Rohrlich correlations;
\item[HS  model] no  entangled  state,  all eight  entangled
  effects, free local dynamics,  so called since it violates
  the no-hypersignaling principle~\cite{DBTBV17};
\item[four   frozen  models]   only   one  entangled   state
  $\Omega_x$  and one  entangled effect  $E_x$, with  $x \in
  \{16, 17, 18, 19\}$,  no non-trivial local dynamics (hence
  the name of the models).
\end{description}
Any  other  composition of  two  squits  is necessarily  not
completely  positive  and  possibly not  even  positive,  as
Janotta's model.  Finally, it is  very easy to see that such
models are  completely positive, by observing  that for each
of them one has
\begin{align*}
  \begin{aligned}
    \Qcircuit      @C=1em      @R=1em      {
      \multiprepareC{1}{\Omega_x}                          &
      \qw\\              \pureghost{\Omega_x}              &
      \multimeasureD{1}{E_x}\\ \multiprepareC{1}{\Omega_x} &
      \ghost{E_x}\\ \pureghost{\Omega_x} & \qw}
  \end{aligned} \; \propto \;
  \begin{aligned}
    \Qcircuit      @C=1em      @R=1em       {
      \multiprepareC{1}{\Omega_x} & \qw\\
      \pureghost{\Omega_x}&\qw}
  \end{aligned} \;,
\end{align*}
as well as
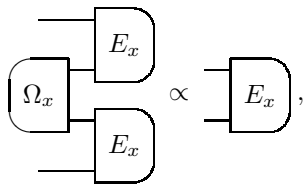
\begin{align*}
  \begin{aligned}
    \Qcircuit      @C=1em      @R=1em      {
      & \multimeasureD{1}{E_x}\\
      \multiprepareC{1}{\Omega_x} & \ghost{E_x}\\
      \pureghost{\Omega_x} & \multimeasureD{1}{E_x}\\
      & \ghost{E_x}
    }
 \end{aligned} \; \propto \;
 \begin{aligned}
    \Qcircuit      @C=1em      @R=1em    {
      & \multimeasureD{1}{E_x}\\
      & \ghost{E_x}}
  \end{aligned} \;,
\end{align*}
and finally
\begin{align*}
  \operatorname{Swap} \Omega_x = \Omega_x,\\
  \operatorname{Swap} E_x = E_x,
\end{align*}
for any $x \in \{16, 17, 18, 19\}$.

\section{Main results}
\noindent
In  this section  we introduce  our main  results, that  is,
analytical    and    algorithmic    techniques    for    the
carachterization  of  all  the  extremal  measurements  with
ray-extremal effects of any given system, as well as for the
computation of its signaling dimension.

\subsection{Extremal measurements with ray-extremal effects}
\label{sec:measurements}
\noindent
It is  well-known~\cite{Par99, DLP05} that  extremal quantum
measurements can comprise effects  that are not ray-extremal
(that is, they are not  rank-one projectors). However, it is
also   known~\cite{Dav78}  that,   whenever  optimizing   an
objective function  that is  convex in the  measurement, one
can  restrict  to  extremal measurements  with  ray-extremal
effects.   In  the  following   we  extend  this  result  to
generalized probabilistic theories.

Before proceeding,  let us  give a convenient  definition of
measurement.   To do  so,  we first  stipulate to  normalize
effects  so  that  any  measurement of  the  system  can  be
expressed as a probability distribution $p$ over the effects
such that $\sum_y  p_y e_y = \overline{e}$,  the unit effect
with unit  probability over any  state. As a  comparison, in
quantum theory such a choice would be equivalent to defining
effects as operators  with trace equal to  the Hilbert space
dimension (the same normalization as the identity operator),
so  that  finite  positive  operator-valued  measures  would
actually  be   uniquely  identified  by   the  corresponding
probability distributions.  Hence, for instance,  in quantum
theory  rank-one   projectors  are  the   only  ray-extremal
effects. Any  other projector, while an  extremal effect, is
not ray-extremal.

The  relevance of  extremal  measurements with  ray-extremal
effects (or,  equivalently, extremal normalized  effects) is
made clear by the following two trivial observations: i) the
maximum of any convex  objective function of the measurement
is attained by an extremal  measurement, and ii) the maximum
of any given  objective function of the  measurement that is
non-decreasing  under   fine  graining  is  attained   by  a
measurement with  extremal normalized  effects.  It  is thus
important   to   characterize    those   measurements   with
ray-extremal effects that are  extremal: this is achieved by
the  following   lemma,  which  was  first   proved  in  the
supplemental material of Ref.~\cite{DBTBV17}.

\begin{lmm}[Characterization of extremal measurements with ray-extremal effects]
  \label{lmm:measurements}
  For any measurement $M = \{ p_y > 0, e_y \}$ with extremal
  normalized effects  $\{ e_y \}$, the  following conditions
  are equivalent:
  \begin{enumerate}
  \item\label{item:extremal} $M$ is extremal,
  \item\label{item:independent}  $\{  e_y \}$  are  linearly
    independent.
  \end{enumerate}
\end{lmm}

\begin{proof}
  Let    us    first    prove    by    contradiction    that
  \ref{item:extremal}  implies  \ref{item:independent}.   We
  start   by  assuming   that  there   exists  an   extremal
  measurement $\{  p_y, e_y \}$ with  $p > 0$ such  that $\{
  e_y \}$ are not linearly independent, i.e.  $| \{ e_y \} |
  > \dim\spn(\{ e_y \})$.  Since $\{ e_y \}$ are normalized,
  they   belong  to   an   affine   subspace  of   dimension
  $\dim\spn(\{ e_y \}) -  1$.  Thus, applying Caratheodory's
  theorem, the  unit effect $\overline{e}$, that  belongs to
  $\spn(\{  e_y \})$  by  hypothesis, can  be decomposed  in
  terms  of  a  subset  of  $\{  e_y  \}$  with  cardinality
  $\dim\spn(\{  e_y \})$.   In other  words, there  exists a
  probability   distribution  $p'_y$,   whose  support   has
  cardinality not  greater than $\dim\spn(\{ e_y  \})$, such
  that $\sum_y p'_y e_y = \overline{e}$.  By taking $\lambda
  > 0$ such  that $p - \lambda  p' \ge 0$ (such  a $\lambda$
  always   exists   since   $p   >   0$)   and   $p''_y   :=
  (1-\lambda)^{-1}  (p_y -  \lambda  p'_y)$, it  immediately
  follows that  also $\{  p''_y, e_y  \}$ is  a measurement.
  Then measurement $M = \{ p_y, e_y \}$ can be decomposed as
  $\lambda \{p'_y,  e_y \}  + (1-\lambda) \{p''_y,  e_y \}$,
  i.e.  it is not extremal, thus leading to a contradiction.

  Let  us  now  prove  that  \ref{item:independent}  implies
  \ref{item:extremal}.  Since $\{ e_y \}$ are extremal, they
  cannot be further decomposed,  so any convex decomposition
  of $M$ would  necessarily involve a subset  of the effects
  $\{ e_y \}$.  Since such  effects $\{ e_y \}$ are linearly
  independent, the decomposition of $\overline{e}$ is unique
  and, since  $p >  0$, no subset  of $\{ e_y  \}$ can  be a
  measurement.  Therefore, the statement follows.
\end{proof}

As an immediate consequence of Lemma~\ref{lmm:measurements},
for any  extremal measurement with ray-extremal  effects one
has that the number $n$ of  such effects is upper bounded by
the  linear   dimension  of   the  system,  a   result  that
generalizes a theorem  by Davies~\cite{Dav78} to generalized
probabilistic theories.

If the extremal normalized effects  are finite in number and
given  by   the  columns   of  matrix  $E$,   a  probability
distribution $p$ is a measurement on the extremal normalized
effects if  and only  if it  satisfies the  following linear
equalities
\begin{align}
  \label{eq:equalities}
  E p = \overline{e},
\end{align}
as well as following the linear inequalities
\begin{align}
  \label{eq:inequalities}
  p \ge 0.
\end{align}
Therefore,  probability  distribution  $p$  is  an  extremal
measurement on  the extremal normalized effects  if and only
if  it  is  a  vertex  of  such  a  polytope.   Notice  that
Eqs.~\eqref{eq:equalities}       and~\eqref{eq:inequalities}
characterize   such  a   polytope   by   giving  its   faces
description.   The   extremal  measurements   with  extremal
normalized effects  can therefore  be found by  passing from
the  faces  description  to   the  vertices  description,  a
standard problem  that can be  solved e.g.  with  the double
description method~\cite{MRTT53}. This  is summarized by the
following proposition.
\begin{thm}
  For any system with a finite number of extremal normalized
  effects given by  the columns of matrix  $E$, the extremal
  measurements  with  extremal  normalized effects  are  the
  vertices of  the polytope given (in  faces description) by
  Eqs.~\eqref{eq:equalities}    and~\eqref{eq:inequalities},
  and can be found by the double description method. 
\end{thm}

As an  application, let  us go back  to our  running example
given by the HS model,  consisting of the composition of two
squits   that  includes   eight   entangled  effects.    The
aforementioned  procedure produces  in  this case  a set  of
$408$   extremal  measurements   with  extremal   normalized
effects.  Such  a set  can be  partitioned according  to the
equivalence class induced  by the reversible transformations
of the system,  so that two such measurements  belong to the
same  class if  and  only if  they are  equivalent  up to  a
reversible transformation. By taking a single representative
for  each equivalence  class,  the reduced  set of  extremal
measurements with  extremal normalized effects is  left with
$15$      elements      only,       as      reported      in
Table~\ref{tab:measurements}.

\begin{table*}[htb!]
  \centering
  \scalebox{.55}{
  \begin{tabular}{| >{$}c<{$} | >{$}c<{$} ? >{$}c<{$} | >{$}c<{$} | >{$}c<{$} | >{$}c<{$} | >{$}c<{$} | >{$}c<{$} | >{$}c<{$} | >{$}c<{$} | >{$}c<{$} | >{$}c<{$} | >{$}c<{$} | >{$}c<{$} | >{$}c<{$} | >{$}c<{$} | >{$}c<{$} | >{$}c<{$} ? >{$}c<{$} | >{$}c<{$} | >{$}c<{$} | >{$}c<{$} | >{$}c<{$} | >{$}c<{$} | >{$}c<{$} |  >{$}c<{$} |}
    \hline
    \bf M &\bf \# & \bf E_0 & \bf E_1 & \bf E_2 & \bf E_3 & \bf E_4 & \bf E_5 & \bf E_6 & \bf E_7 & \bf E_8 & \bf E_9 & \bf E_{10} & \bf E_{11} & \bf E_{12} & \bf E_{13} & \bf E_{14} & \bf E_{15} & \bf E_{16} & \bf E_{17} & \bf E_{18} & \bf E_{19} & \bf E_{20} & \bf E_{21} & \bf E_{22} & \bf E_{23}\\
    \hline\hline
    0 & 2 &&&&&&&&&&&&&&&&& 120 && 120 &&&&&\\
    \hline\hline
    1 & 4 &60&&60&&&&&&60&&60&&&&&&&&&&&&&\\
    \hline
    2 & 4 &60&&60&&&&&&&60&&60&&&&&&&&&&&&\\
    \hline\hline
    3 & 6 &30&30&&&&&&&&&30&30&&&&&&&60&&&&&60\\
    \hline
    4 & 6 &30&&&&&30&&&&&30&&&&&30&&&&&60&&&60\\
    \hline
    5 & 6 &40&&&&&&&&&&40&&&&&&&40&40&&40&&&40\\
    \hline\hline
    6 & 7 &30&30&&&&&30&&30&&30&&&&&30&&&&&&&&60\\
    \hline\hline
    7 & 8 &20&20&&&20&&&&&&40&&&&&20&&&40&&40&&&40\\
    \hline
    8 & 8 &20&20&&&&&40&&20&20&&&&&&40&40&&&&&&&40\\
    \hline
    9 & 8 &40&20&&&&&20&&&20&&40&&&20&&&&40&&&&&40\\
    \hline
    10 & 8 &30&&&&&30&&&&&&30&&&30&&&&30&30&30&&&30\\
    \hline\hline
    11 & 9 &20&20&&&20&&20&&&20&20&&&&&40&&&&&40&&&40\\
    \hline
    12 & 9 &15&15&&&15&&30&&&30&&&&&&45&30&&&&30&&&30\\
    \hline
    13 & 9 &20&20&&&20&&&20&&&20&20&&20&20&&&&80&&&&&\\
    \hline
    14 & 9 &24&&24&&&24&&&&&&48&&24&&&&&24&24&&&24&24\\
    \hline\hline
    \bf M &\bf \# & \bf E_0 & \bf E_1 & \bf E_2 & \bf E_3 & \bf E_4 & \bf E_5 & \bf E_6 & \bf E_7 & \bf E_8 & \bf E_9 & \bf E_{10} & \bf E_{11} & \bf E_{12} & \bf E_{13} & \bf E_{14} & \bf E_{15} & \bf E_{16} & \bf E_{17} & \bf E_{18} & \bf E_{19} & \bf E_{20} & \bf E_{21} & \bf E_{22} & \bf E_{23}\\
    \hline
  \end{tabular}}
  \tcaption{Set   of  all   extremal  measurements   (up  to
    reversible transformations)  for the composition  of two
    squits  named HS  model  with  eight entangled  effects.
    Measurements are  labelled by  $M$ and represented  by a
    probability distribution  $p$ over  the set  of extremal
    normalized effects $\{ E_j \}$  (rescaled by a factor of
    $240$ so that each entry  is an integer).  The number of
    non-null values of $p$  is also reported for convenience
    in the column  indicated by the symbol  $\#$.  We recall
    that factorized effects are those from $E_0$ to $E_{15}$
    (included), while the other effects are entangled.}
  \label{tab:measurements}
\end{table*}

\subsection{The signaling dimension}
\label{sec:signalingdimension}
\noindent
We  have  already  shown  that,  in  order  to  compute  the
signaling  dimension of  any  given system,  it suffices  to
consider  the  $m$-input/$n$-output conditional  probability
distributions  generated by  the extremal  measurements with
ray-extremal effects upon the input of all the states of the
system.   Here and  in  the  following we  take  any such  a
conditional probability  distribution $p$ arranged as  an $m
\times n$ stochastic matrix, where $x$ and $y$ label the row
and  the  column, respectively.   For  any  such a  $p$,  it
suffices  to  compute  the  minimal  dimension  $d$  of  the
classical polytope  $\mathcal{P}^{m \to n}_d$  that contains
$p$.  The maximum  over such $p$'s of  the minimal dimension
$d$ constitutes the signaling dimension of the given system.

The straightforward way of achieving  this would be to check
$p$  against  all  the inequalities  that  characterize  the
facets  of   $\mathcal{P}^{m  \to  n}_d$.    However,  known
algorithms for  the characterization of such  facets require
as input  the vertices  of $\mathcal{P}^{m \to  n}_d$, which
are  too  many  to  make  this  approach  feasible  for  the
instances of the problem we are interested in. This includes
standard   algorithms  such   as   the  double   description
method~\cite{MRTT53}, as well  as fine-tuned algorithms such
as the  adjacency decomposition  algorithm by  Doolittle and
Chitambar   of  Ref.~\cite{DC21},   that  can   exploit  the
symmetries of  $\mathcal{P}^{m \to n}_d$ in  order to reduce
the  number   of  facets   they  output  by   producing  one
representative  facet  for   each  equivalence  class  under
symmetry.   For  instance,  for  $m =  16$  (the  number  of
extremal states in the HS composition of the squit), $n = 9$
(the maximum  number of  elements of  extremal measurements,
attained    by    four     such    measurements    as    per
Table~\ref{tab:decompositions}),  and $d  = 5$  (the minimum
value we are interested in),  the number of such vertices is
$\sim  2.4  \cdot   10^{13}$  (see  Lemma~\ref{lmm:vertices}
below).   From  this  issue,   the  need  arises  to  devise
techniques that do not  requrire the explicit enumeration of
all the  vertices of $\mathcal{P}^{m  \to n}_d$, as  done in
the following.

For any  given conditional probability distribution
$p$ and any classical dimension  $d$, in order to prove that
$d$ is  the minimum dimension  such that $p$ belongs  to the
classical polytope  $\mathcal{P}^{m \to n}_d$, one  needs to
provide:
\begin{itemize}
\item an  explicit convex decomposition $\mathbf{x}$  of $p$
  in terms of the vertices of $\mathcal{P}^{m \to n}_d$, and
\item  an explicit  linear witness  (or game,  or separating
  hyperplane) $g$ such that $p \cdot  g > q \cdot g$ for any
  $q \in \mathcal{P}^{m \to n}_{d - 1}$.
\end{itemize}

To  obtain a  convex  decomposition $\mathbf{x}$  of $p$  in
terms of the vertices of $\mathcal{P}^{m \to n}_d$, whenever
it exists, one  can proceed as follows.  The  problem can be
framed as the feasibility of the following linear program
\begin{align}
  \label{eq:primal}
  \min_{\substack{A \mathbf{x}  = \mathbf{b}\\\mathbf{x} \ge
      0}} \mathbf{c} \cdot \mathbf{x},
\end{align}
where $A$  is the $(mn)  \times V$ matrix whose  columns are
the $V$ vertices of  $\mathcal{P}^{m \to n}_d$ rearranged as
vectors   with  $m   \times  n$   entries  (the   particular
rearranging  is  irrelevant  as long  as  used  consistently
across the protocol); $\mathbf{b}$  is the vector with $mn$
entries  obtaining rearranging  the conditional  probability
distribution  $p$;  $\mathbf{c}$  is  the  vector  with  $V$
entries all  equal to  zero since, as  said earlier,  we are
interested  in  the  feasibility  of the  problem  only  (an
alternative option is to take $\mathbf{c}$ to be the vectors
with all entries  equal to one, which also  gives a constant
objective function since $\mathbf{x}$ is constrained to be a
probability distribution).

In turn, the  number $V$ of vertices  of $\mathcal{P}^{m \to
  n}_d$ (equivalently, the number  of columns of matrix $A$)
can itself  be expressed in terms  of $m$, $n$, and  $d$. To
see this, let us denote with
\begin{align*}
  \binom{n}{k} := \frac{n!}{k!(n-k)!}
\end{align*}
the  binomial   coefficient  and  with
\begin{align*}
  {m \brace           k}           :=           \sum_{j=0}^{k}
  \frac{1}{k!}(-1)^{k-j}\binom{k}{j} j^m
\end{align*}
the Stirling number of the  second kind, i.e.  the number of
partitions  of  a  set  of $m$  elements  in  $k$  non-empty
classes.   Then the  following result,  first proved  in the
supplemental material of Ref.~\cite{DBTBV17}, holds.

\begin{lmm}
  \label{lmm:vertices}
  The number $V$ of vertices of $\mathcal{P}^{m \to n}_d$ is
  given by
  \begin{align*}
    V = \sum_{k=1}^{d} k! \binom{n}{k} {m\brace k}.
  \end{align*}
\end{lmm}

\begin{proof}
  The  statement follows  by a  a simple  counting argument.
  First, one chooses $k \leq  d$ non-null columns (as stated
  earlier,   different  columns   correspond  to   different
  effects); since  the matrix has  a $n$ columns,  there are
  $\binom{n}{k}$ such possible choices.  Then, for each such
  a choice one has to assign a single one to each row within
  the $k$ columns that were  chosen; in other words, one has
  to assign each  row (playing here the role  of the element
  of a set) to a column (playing here the role of an element
  of a partition of such a  set).  There are exactly $k!  {m
    \brace  k}$ possible  assignments,  where the  factorial
  $k!$ comes  from the  fact that here  the elements  of the
  partition  are  labeled,  while   the  definition  of  the
  Stirling number of the second kind does not take this into
  account.
\end{proof}

In typical instances of the problem (see our running example
on the composition  of two squits, discussed  later on), $V$
is too large for matrix  $A$ to be practically tractable. In
this case, the following steps  can reduce the complexity of
the problem.  First, observe that  any row of $p$ (as stated
earlier, different rows correspond to different states) that
is the  convex combination of  other rows can  be eliminated
without  altering the  result, thus  reducing the  effective
value  of $m$  (and thus  $V$) without  loss of  generality.
Also,  in typical  applications the  conditional probability
distributions $p$'s  are rather sparse,  a fact that  can be
exploited  as follows.   Any vertex  of $\mathcal{P}^{m  \to
  n}_d$  that  contains an  entry  equal  to one  where  $p$
contains  a   zero  will   not  contribute  to   the  convex
decomposition of $p$; hence can be discarded without loss of
generality, thus helping to further reduce the number $V$ of
vertices.   In the  following,  we denote  with  $v$ such  a
reduced number of vertices, and we refer to it as the number
of   effective  vertices.   We  have   then  the   following
proposition.
\begin{thm}
  Any   given   $m   \times   n$   conditional   probability
  distribution $p$ belongs to the polytope $\mathcal{P}_d^{m
    \to  n}$ if  and only  if  $q$ belongs  to the  polytope
  $\mathcal{Q}_d^{m' \to n}$, where $q$ is $m' \times n$ and
  is  the same  as  $p$ where  any row  that  is the  convex
  combination of rows is removed, and $\mathcal{Q}_d^{m' \to
    n}$ is the same as  $\mathcal{P}_d^{m' \to n}$ where any
  vertex that  has an entry  equal to  one where $q$  has an
  entry equal to zero has been removed.
\end{thm}

The process to reduce the number of vertices from $V$ to $v$
can be  efficiently implemented  as a  combinatorial, branch
and bound algorithm.  The algorithm  starts on the first row
of matrix $p$, cycles over any choice of non-null entries of
that row,  and for each  choice recursively calls  itself on
the  second  row.   At  each call,  the  bounding  procedure
consists of verifying if the number of columns from which an
entry has been chosen so far  is larger than $d$; if so, the
branch is pruned. Any  pre-determined branching strategy can
be adopted.

If the above steps do not suffice to make the linear program
in  Eq.~\ref{eq:primal} tractable,  the  technique known  as
delayed column  generation~\cite{PS82, Chv83,  GLS88, Ala99}
can  be  adopted,  by  observing  that  it  is  possible  to
efficiently generate  the vertex  of the  classical polytope
$\mathcal{P}^{m \to  n}_d$ that minimizes the  inner product
with any  given vector;  this immediately gives  the reduced
cost to be used in delayed column generation.

To obtain a linear witness  $g$ that separates the classical
polytope  $\mathcal{P}^{m \to  n}_d$ from  $p$, whenever  it
exists,  one can  proceed as  follows.  The  problem can  be
framed as the  following linear program, dual of  the one in
Eq.~\ref{eq:primal} except  for the inclusion of  a bounding
box around variable $\mathbf{y}$:
\begin{align}
  \label{eq:dual}
  \max_{\substack{A^T  \mathbf{y} =  \mathbf{c}\\-\mathbf{u}
      \le  \mathbf{y}  \le   \mathbf{u}}}  \mathbf{b}  \cdot
  \mathbf{y},
\end{align}
where matrix  $A$ and vectors $\mathbf{b}$  and $\mathbf{c}$
are defined as above, and  vector $\mathbf{u}$ is the vector
with  $mn$  entries  all  equal to  one.   Notice  that,  if
$\mathbf{c}$  has  been taken  to  be  the vector  with  all
entries equal to one, then one  does not need to include the
bounding box condition. The  witness (or game, or separating
hyperplane) $g$ is thus recovered by rearranging the entries
of vector $\mathbf{y}$ as a matrix.

As    noticed   above    for   the    primal   problem    in
Eq.~\eqref{eq:primal},      the     dual      problem     in
Eq.~\eqref{eq:dual} too  is typically intractable,  that is,
there  are  too  many   constraints.   In  addition  to  the
techniques discussed above to reduce the size of matrix $A$,
one can in this  case adopt the ellipsoid method~\cite{PS82,
  Chv83,    GLS88,    Ala99}    or   the    cutting    plane
method~\cite{PS82,  Chv83,  GLS88,  Ala99}, by  using  as  a
separation   oracle   the   aforementioned   function   that
efficiently  returns the  vertex of  the classical  polytope
$\mathcal{P}^{m  \to n}_{d  - 1}$  that minimizes  the inner
product with any given vector.

As an application of the  discussion of this section, let us
return to  our running  example by considering  the extremal
measurements with  ray-extremal effects for  any composition
of   two  squits   given  in   Table~\ref{tab:measurements}.
Considering the composition named HS-model that includes all
eight entangled effects (and  therefore no entangled state),
upon the input of the $16$ (factorized) extremal states each
such  measurement gives  rise to  a conditional  probability
distribution with $16$ rows and a number of columns equal to
the  number of  effects.   Notice that,  when computing  the
signaling  dimension,  it   suffices  to  consider  extremal
measurements with  at least four ray-extremal  effects, that
is,    measurements   number    $3$    up    to   $14$    in
Table~\ref{tab:measurements}.

The two  steps discussed above  (that is, the  derivation of
the convex decomposition of $p$ or of the separating witness
$g$) have  been implemented in  function \verb|decompose.m|,
and      the       results      are       summarized      in
Table~\ref{tab:decompositions}.
\begin{table}[htb!]
  \centering
  \begin{tabular}{| >{$}c<{$} | >{$}c<{$} ? >{$}c<{$} | >{$}c<{$} |  >{$}c<{$} |  >{$}c<{$} |}
    \hline
    \bf M &\bf \# & \bf d & \bf g \cdot b & \bf v & \bf V\\
    \hline\hline
    3 & 6 & 4 & & 128 & \sim 9 \cdot 10^{10}\\
    \hline
    4 & 6 & 4 & & 64 & \sim 9 \cdot 10^{10}\\
    \hline
    5 & 6 & 4 & & 465 & \sim 9 \cdot 10^{10}\\
    \hline\hline
    6 & 7 & 5 & 2 & 672 & \sim 4 \cdot 10^{12}\\
    \hline\hline
    7 & 8 & 5 & \nicefrac13 & 60752 & \sim 10^{13}\\
    \hline
    8 & 8 & 5 & \nicefrac83 & 7616 & \sim 10^{13}\\
    \hline
    9 & 8 & 5 & 2 & 10040 & \sim 10^{13}\\
    \hline
    10 & 8 & 4 & & 576 & \sim 3 \cdot 10^{11}\\
    \hline\hline
    11 & 9 & 5 & \nicefrac43 & 37136 & \sim 2 \cdot 10^{13}\\
    \hline
    12 & 9 & 5 & 2 & 107504 & \sim 2 \cdot 10^{13}\\
    \hline
    13 & 9 & 5 & \nicefrac23 & 8704 & \sim 2 \cdot 10^{13}\\
    \hline
    14 & 9 & 5 & \nicefrac85 & 488092 & \sim 2 \cdot 10^{13}\\
    \hline\hline
    \bf M &\bf \# & \bf d & \bf g \cdot b & \bf v & \bf V\\
    \hline
  \end{tabular}
  \tcaption{Set   of  all   extremal   measurements  as   in
    Table~\ref{tab:measurements}.   The two  columns labeled
    with $d$  and $\mathbf{g}  \cdot \mathbf{b}$  denote the
    minimal   dimension  $d$   of  the   classical  polytope
    $\mathcal{P}^{m  \to  n}_d$  such that  the  conditional
    probability distribution $p$ obtained by the measurement
    upon the  input of the  $16$ extremal states  belongs to
    $\mathcal{P}^{m \to  n}_d$, and (whenever such  a $d$ is
    larger  than  $4$) the  maximum  value  attained by  any
    witness $g$  separating $p$ from the  classical polytope
    $\mathcal{P}^{m \to n}_{d -  1}$, respectively.  The two
    columns labeled with $v$ and $V$ represent the number of
    effective vertices for such  a measurement and the total
    number   of   vertices   of   the   classical   polytope
    $\mathcal{P}^{m  \to n}_d$.   Even  in  the worst  case,
    corresponding to the last row of the table, the protocol
    described  here reduces  the size  of the  problem by  a
    factor $V/v \sim 4 \cdot 10^7$.}
  \label{tab:decompositions}
\end{table}
Hence, the following corollary follows.
\begin{cor}
  The signaling  dimension of the composition  of two squits
  named  HS model,  including all  eight possible  entangled
  effects, is five.
\end{cor}
Incidentally, our results proves  the tightness of the lower
bound on  the signaling dimension  of such a model  given in
Ref.~\cite{DBTBV17}.

\section{Conclusion}
\label{sec:conclusion}
\noindent
In  this   work  we   derived  analytical   and  algorithmic
techniques  to characterize  the  extremal measurements  and
compute the signaling  dimension of any given  system of any
given generalized  probabilistic theory.  As an  example, we
applied our results  to the composition of  two square bits.
The  algorithmic  techniques  we  derived  here,  and  whose
implementation is made available online~\cite{Dal23}, can be
directly applied  to other  system whose states  and effects
form polytopes, such as polygonal and polyhedral theories.

\section*{Acknowledgments}
\noindent
M.~D.  acknowledges support from  the Department of Computer
Science and Engineering, Toyohashi University of Technology,
from the International Research Unit of Quantum Information,
Kyoto  University, and  from the  JSPS KAKENHI  grant number
JP20K03774.   A.~T.  acknowledges  support from  the Silicon
Valley   Community   Foundation   Project   ID\#2020-214365.
F.~B. acknowledges  support from MEXT-JSPS  Grant-in-Aid for
Transformative Research Areas  (A) ``Extreme Universe,'' No.
21H05183;  from MEXT  Quantum  Leap  Flagship Program  (MEXT
QLEAP)  Grant No.   JPMXS0120319794; and  from JSPS  KAKENHI
Grants No. 20K03746 and No. 23K03230.

\end{document}